\documentclass[draftcls,onecolumn]{IEEEtran}

\usepackage[utf8]{inputenc}
\usepackage[T1]{fontenc}
\usepackage{url}
\usepackage{ifthen}
\usepackage{cite}
\usepackage[cmex10]{amsmath} 
\usepackage{graphicx}
\usepackage{amsmath, amsthm, amssymb, amsfonts}
\usepackage{tabularx}
\usepackage[linesnumbered,ruled]{algorithm2e}
\usepackage{algpseudocode}
\usepackage{subfigure}
\usepackage{url}
\usepackage{xcolor}
\usepackage{multirow}
\usepackage{colortbl}
\newtheorem{definition}{Definition}
\newtheorem{lemma}{Lemma}

\newtheorem{theorem}{Theorem}
\newtheorem{remark}{Remark}

\begin{document}
\title{Generalized Universal Coding of Integers}

\author{Wei Yan,
        Sian-Jheng Lin,~\IEEEmembership{Member,~IEEE}
        and Yunghsiang S. Han,~\IEEEmembership{Fellow,~IEEE,}
\thanks{This work was partially presented at the 2021 IEEE Information Theory Workshop.

W. Yan is with the School of Cyber Science and Technology, University of Science and Technology of China (USTC), Hefei, Anhui, China, email: yan1993@mail.ustc.edu.cn}
\thanks{S.-J. Lin is with the Huawei Technology Co., Ltd, Hong Kong, China, email:lin.sian.jheng1@huawei.com}
\thanks{Y. S. Han is with Shenzhen Institute for Advanced Study, University of Electronic Science and Technology of China, Shenzhen, China, and also with the College of Electrical Engineering and Computer Science, National Taipei University, Taipei, Taiwan, email: yunghsiangh@gmail.com.
}
}
\maketitle
\begin{abstract}
Universal coding of integers~(UCI) is a class of variable-length code, such that the ratio of the expected codeword length to $\max\{1,H(P)\}$ is within a constant factor, where $H(P)$ is the Shannon entropy of the decreasing probability distribution $P$. However, if we consider the ratio of the expected codeword length to $H(P)$, the ratio tends to infinity by using UCI, when $H(P)$ tends to zero. To solve this issue, this paper introduces a class of codes, termed generalized universal coding of integers~(GUCI), such that the ratio of the expected codeword length to $H(P)$ is within a constant factor $K$.
First, the definition of GUCI is proposed and the coding structure of GUCI is introduced.
Next, we propose a class of GUCI $\mathcal{C}$ to achieve the expansion factor $K_{\mathcal{C}}=2$ and show that the optimal GUCI is in the range $1\leq K_{\mathcal{C}}^{*}\leq 2$.
Then, by comparing UCI and GUCI, we show that when the entropy is very large or $P(0)$ is not large, there are also cases where the average codeword length of GUCI  is shorter.
Finally, the asymptotically optimal GUCI is presented.
\end{abstract}
\IEEEpeerreviewmaketitle

\section{Introductions}\label{sec_intro}
For lossless source coding, there are three major categories, termed variable-to-fixed length~(VF) codes (e.g. Tunstall code~\cite{T67}), fixed-to-variable length~(FV) codes (e.g. Huffman code~\cite{H52}) and variable-to-variable length~(VV) codes (e.g. Khodak code~\cite{K72,VV08}).
As their name implies, VF codes encode a variable-length sequence of source symbols into a constant-length codeword, and FV codes encode a constant-length sequence of source symbols into a variable-length codeword. In particular, variable-length codes map the source symbols to a variable number of bits, and this is the most important type of FV codes.
VF and FV codes are the special cases of VV codes, and the main research on VV codes focuses on redundancy rates~\cite{K72,VV08,VV02,VV04}.

In particular, universal coding of integers~(UCI) is a variable-length code for the discrete memoryless sources with the infinite alphabet, and the probability distribution of sources does not need prior knowledge. 
In 1968, Levenshtein~\cite{L68} proposed the first UCI, although UCI was not yet defined then. 
In 1975, Elias~\cite{Elias75} established the fundamental framework of UCI. Elias considered a discrete memoryless source $S=(P,\mathcal{A})$ with a countable alphabet set $\mathcal{A}\triangleq\mathbb{N}^{+}=\{1,2,3,\cdots\}$
and a decreasing probability distribution (DPD) $P$ of $\mathbb{N}^{+}$ (i.e., $\sum_{n=1}^{\infty}P(n)=1$, and $P(m)\geq P(m+1)\geq 0$, for all $m\in \mathbb{N}^{+}$).
Let $H(P)=-\sum_{n=1}^{\infty}P(n)\log_{2}P(n)$ denote the Shannon entropy of $P$.
Let $\mathcal{C}$ be a variable-length code for the source $S=(P,\mathbb{N}^{+})$, and it maps the positive integers $\mathbb{N}^{+}$ onto binary codewords $\{0,1\}^{*}$.
Let $L_{\mathcal{C}}(\cdot)$ denote the length function so that $L_{\mathcal{C}}(m)=|\mathcal{C}(m)|$, 
for all $m\in \mathbb{N}^{+}$.
Furthermore, $E_{P}(L_{\mathcal{C}})=\sum_{n=1}^{\infty}P(n)L_{\mathcal{C}}(n)$ denotes the expected codeword length of $\mathcal{C}$. We say that $\mathcal{C}$ is \emph{universal} if
\begin{equation}\label{e1}
\frac{E_{P}(L_{\mathcal{C}})}{\max\{1,H(P)\}}\leq K_{\mathcal{C}},
\end{equation}
for all DPD $P$ with $H(P)<\infty$.
$K_{\mathcal{C}}$ is termed \emph{expansion factor} of UCI $\mathcal{C}$ and $K_{\mathcal{C}}^{*}\triangleq\inf\{K_{\mathcal{C}}\}$ is termed \emph{minimum expansion factor} of UCI $\mathcal{C}$.
Moreover, $\mathcal{C}$ is called \emph{asymptotically optimal} if $\mathcal{C}$ is universal
and there exists a function $R_{\mathcal{C}}(\cdot)$ such that
\begin{equation}\label{eq6}
  \frac{E_{P}(L_{\mathcal{C}})}{\max\{1,H(P)\}}\leq R_{\mathcal{C}}(H(P)),
\end{equation}
for all DPD $P$ with $H(P)<\infty$ and
\[
  \lim\limits_{H(P)\to +\infty}R_{\mathcal{C}}(H(P))=1.
\]

UCI has two main categories~\cite{C1990}, namely \emph{message length strategy} and \emph{flag pattern strategy}.
The $\gamma$ code, $\delta$ code and $\omega$ code, proposed by Elias~\cite{Elias75}, belong to the message length strategy.
This strategy UCI is mainly the recursive code to minimize $L_{\mathcal{C}}(m)$ for large $m\in\mathbb{N}$.
For example, two classes of UCIs proposed by Stout~\cite{S80} to improve $\omega$ code for large $m$.
For another example, Yamamoto~\cite{Y00} cleverly designed a delimiter with a length greater than $1$ to construct a new class of UCI whose length function satisfies
\[
  L_{\mathcal{C}}(m)<\log_{2}m+\log_{2}(\log_{2}m)+\cdots+\log_{2}^{t^{*}(m)}m,
\]
where $t^{*}(m)$ is the largest positive integer $t$ satisfying $\log_{2}^{t}m\geq0$.
However, the message length strategy UCI shall be used in an error-free environment.
Instead, the flag pattern strategy UCI, first studied by Lakshmanan~\cite{L81}, just make up for the problem, it has certain resynchronization properties.
Families of Fibonacci codes~\cite{AF87} is probably the most famous flag pattern strategy UCI, but it is not asymptotically optimal and the encoding and decoding are complex.
A new flag pattern strategy UCI, proposed by Wang~\cite{W88}, has improved in the above two aspects.
Yamamoto~\emph{et al.}~\cite{YO91} further improved and promoted Wang's coding scheme.
Furthermore, Amemiya~\emph{et al.}~\cite{A1993} provided a new \emph{group strategy} UCI.
The message length strategy coding can be regarded as a special group strategy coding.

Recently, \'{A}vila~\emph{et al.}~\cite{AS17} proposed a new family of UCI whose length function can reach the bounds  in~\cite{L68,AHK97,LC78} in the sense of a difference of a constant.
The work of Allison~\emph{et al.}~\cite{DCC21} focused on the universality of Wallace tree code. 
Yan \emph{et al.}~\cite{YL21,YL22} first studied the range of the minimum expansion factor of UCI.
If a class of UCI $\mathcal{C}$ has the smallest minimum expansion factor $K_{\mathcal{C}}^{*}$, then $\mathcal{C}$ is termed \emph{optimal UCI}.
The authors proved that the optimal UCI is in the range $2\leq K_{\mathcal{C}}^{*}\leq 2.5$, where $K_{\mathcal{C}}^{*}=2.5$ is achieved by $\iota$ code~\cite{YL22}. Today, UCI is used in many applications, such as biological sequencing data compression~\cite{DNA10,DNA13}, inverted file index~\cite{06} and unbounded search problems~\cite{BY76,AHK97}.

However, from \eqref{e1}, for a universal code, the ratio of $E_{P}(L_{\mathcal{C}})$ to $H(P)$ cannot be within a constant $K_{\mathcal{C}}$ when $H(P)$ approaches zero.
That is, UCI cannot satisfy the following inequality.
\begin{equation}\label{e2}
	\frac{E_{P}(L_{\mathcal{C}})}{H(P)}\leq K_{\mathcal{C}}.
\end{equation}
Precisely, when $H(P)$ is extremely small, the expected codeword length of a variable-length code is
\[
  E_{P}(L_{\mathcal{C}})=\sum_{n=1}^{\infty}P(n)L_{\mathcal{C}}(n)\geq \sum_{n=1}^{\infty}P(n)\cdot 1=1.
\]
Clearly, the left-hand side of \eqref{e2} tends to infinity when $H(P)$ tends to $0$.
Thus, a traditional UCI cannot meet the inequality \eqref{e2}, and the objective of this work is to propose a new class of code satisfying the inequality similar to \eqref{e2}. That is, the motivation of this work is to perfect the theoretical definition so that regardless of the size of the source entropy, there is a code whose average code length is less than a constant multiple of the entropy.

In this paper, we introduce a family of codes, called generalized universal coding of integers~(GUCI), which is a generalization of UCI via VV codes. In particular, GUCI can meet the inequality similar to \eqref{e2}.
The minimum expansion factor of GUCI is also studied. The major contributions are enumerated as follows.
\begin{enumerate}
\item The definition of GUCI and asymptotically optimal GUCI is presented.
\item A family of GUCI and asymptotically optimal GUCI is proposed.
\item In the proposed family of GUCI, a class of GUCI is proposed to achieve the expansion factor $2$. We also show that the optimal GUCI is in the range $1\leq K_{\mathcal{C}}^{*}\leq 2$.
\item The relationship between UCI and GUCI is discovered.
\item A sufficient condition for the average codeword length of GUCI to be shorter than UCI is obtained.
In addition, when Shannon entropy $H(P)$ is large or $P(0)$ is not large, 
there are still cases where the average codeword length of GUCI  is shorter. 
\end{enumerate}

In the rest of this paper,
Section \ref{pre} provides some background knowledge.
Section \ref{sec_lb} defines GUCI.
A family of GUCI is provided in Section \ref{sec_new}.
Section \ref{factor} discusses the expansion factor of GUCI.
Section \ref{length} compares the average code length of this family of GUCI and the original UCI.
Section \ref{sec_dis} studies the definition and property of asymptotically optimal GUCI.
Section \ref{sec_con} concludes this work.

\section{Preliminaries}\label{pre}
\subsection{Notations}
Let $\mathbb{N} \triangleq \{0\}\bigcup \mathbb{N}^{+}$ denote the set of non-negative numbers. Let $\alpha(m)$ denote the unary representation of the positive number $m$. For example, $\alpha(1)=1$, $\alpha(2)=01$ and $\alpha(5)=00001$. Let $\beta(m)$ denote the standard binary representation of a positive integer $m$. Let $[\beta(m)]$ denote the binary code by removing the most significant bit $1$ of $\beta(m)$. For example, $\beta(9)=1001$ and $[\beta(9)]=001$. Then, we obtain
\begin{equation*}
\begin{aligned}
 |\alpha(m)|&= m,  \\
  |\beta(m)|&= 1+\lfloor \log_{2}m\rfloor,  \\
 |[\beta(m)]|&= \lfloor \log_{2} m\rfloor,
\end{aligned}
\end{equation*}
for all $m\in \mathbb{N}^{+}$.

\subsection{Elias $\gamma$ code and the codeword lengths of some classical UCIs}
As Elias $\gamma$ code is frequently used in this paper, here is a detailed introduction to the specific structure of Elias $\gamma$ code. For other classic UCIs, please refer to~\cite{Elias75,YL21,YL22}.
Elias $\gamma$ code was introduced by Elias~\cite{Elias75}.
It is an encoding scheme for message length.
Elias $\gamma$ code: $\mathbb{N}^{+}\rightarrow \{0,1\}^{*}$ can be expressed as
\[
	\gamma(m)=\alpha(|\beta(m)|)[\beta(m)],
\]
for all $m\in \mathbb{N}^{+}$. The role of the leading $0'$s is to ensure that Elias $\gamma$ code is a prefix code.
The codeword length is given by
\begin{equation*}
\begin{aligned}
			|\gamma(m)|&=|\alpha(1+\lfloor \log_{2} m\rfloor)|+|[\beta(m)]|    \\
			&=1+\lfloor \log_{2} m\rfloor+\lfloor \log_{2} m\rfloor          \\
			&=1+2\lfloor \log_{2}m\rfloor.
\end{aligned}
\end{equation*}
For example, $\gamma(9)=0001001$ and $|\gamma(9)|=1+2\lfloor \log_{2}9\rfloor=7$.
Elias $\gamma$ code is universal, but it is not asymptotically optimal.
Next, a lemma about the length function of other classic UCIs is given.
\begin{lemma}~\cite{Elias75,YL21,YL22}
The following classic UCIs have $L_{\mathcal{C}}(1)=1$.
For all $2\leq n\in\mathbb{N}^{+}$,
\begin{enumerate}
\item the length function of $\delta$ code satisfies $L_{\delta}(n)=1+\lfloor\log_{2}n\rfloor+2\lfloor\log_{2}(1+\lfloor\log_{2}n\rfloor)\rfloor$;
\item the length function of $\eta$ code satisfies $L_{\eta}(n)= 3+\lfloor\log_{2}(n-1)\rfloor+\lfloor\frac{\lfloor\log_{2}(n-1)\rfloor}{2}\rfloor$;
\item the length function of $\theta$ code satisfies $L_{\theta}(n)= 3+\lfloor\log_{2}n\rfloor+\lfloor\log_{2}\lfloor\log_{2}n\rfloor\rfloor+\lfloor\frac{\lfloor\log_{2}\lfloor\log_{2}n\rfloor\rfloor}{2}\rfloor$;
\item the length function of $\iota$ code satisfies $L_{\iota}(n)= 2+\lfloor\log_{2}n\rfloor+\lfloor\frac{1+\lfloor\log_{2}n\rfloor}{2}\rfloor$;
\item the length function of $\omega$ code satisfies $L_{\omega}(n)=1+\sum_{m=1}^{t}(1+\lambda^{m}(n))$, where $\lambda(n)\triangleq\lfloor\log_{2}n\rfloor$, $\lambda^{m}$ denotes the $m$-fold compositions of $\lambda$,
and $t=t(n)\in\mathbb{N}^{+}$ is a uniquely integer satisfying $\lambda^{t}(n)=1$. Furthermore, $L_{\omega}(n)\leq 3+2\lfloor\log_{2}n\rfloor$.
\end{enumerate}
\end{lemma}

\subsection{Run-length encoding}
Run-length encoding~(RLE)~\cite{Capon59} is essentially a method of encoding run-length rather than encoding individual values. For example, a scan line, consisting of black pixels $B$ and white pixel $W$, may read as follows.
\begin{equation*}
\begin{aligned}
  & WWWWWWWBBBWWWWBWWWWW    \\
  & WWWWWWWWBBWWWWWWWWWW.
\end{aligned}
\end{equation*}
With RLE algorithm, it is encoded as
\[
 7W3B4W1B13W2B10W.
\]
Moreover, RLE can be modified to accommodate data properties.
For instance, the above scan line can also be encoded as
\[
(W,7,3,4,1,13,2,10),
\]
where the numbers can be encoded by prefix coding.

\subsection{Variable-to-fixed length codes}
The VF codes can be divided into two parts, termed parser and string encoder.
First, the parser partitions the source sequence into a concatenation of variable-length strings.
Each variable-length string belongs to a dictionary $\mathcal{D}$, which contains a set of strings.
Next, the string encoder maps the variable-length string $\alpha\in\mathcal{D}$ into the fixed-length string.
To ensure the completeness and uniqueness of the segmentation of the source sequence, $\mathcal{D}$ is required to be proper and complete.
\begin{definition}~\cite{FK72}
\begin{enumerate}
\item If every variable-length string $\alpha_{i}\in\mathcal{D}$ is not a prefix of another variable-length string $\alpha_{j}\in\mathcal{D}$, then $\mathcal{D}$ is termed proper.
\item If every infinite sequence has a prefix in $\mathcal{D}$, then $\mathcal{D}$ is termed complete.
\end{enumerate}
\end{definition}
For example, a proper and complete dictionary over $\{0,1\}$ is $\mathcal{D}=\{1,01,001,000\}$.

\subsection{Variable-to-variable length codes}
VV codes can be considered as a concatenation of VF codes and FV codes~\cite{VV08,VV02,VV04}. First, the VF encoder maps the variable-length string $\alpha\in\mathcal{D}$ into the fixed-length string, and then the FV encoder maps the fixed-length string into the variable-length string. 
Nishiara \emph{et al.}~\cite{VV00} define the almost surely complete~(ASC) dictionary and the corresponding VV code rate.
\begin{definition}~\cite{VV00}
\begin{enumerate}
\item If the probability that dictionary $\mathcal{D}$ has a prefix of the infinite sequence is $1$, then $\mathcal{D}$ is termed almost surely complete.
\item Let $\mathcal{C}$ be a VV code with a proper and ASC dictionary $\mathcal{D}$ and a VV encoder $\varphi$. Then the coding rate of $\mathcal{C}$ is
\[
  R_{\mathcal{C}}=\frac{\sum_{\alpha\in\mathcal{D}}P(\alpha)|\varphi(\alpha)|}{\sum_{\alpha\in\mathcal{D}}P(\alpha)|\alpha|}.
\]
\end{enumerate}
\end{definition}
An example of a dictionary $\mathcal{D}$ over $\{0,1\}$ that proper and ASC is $\mathcal{D}=\{1,01,001,0001,\cdots\}$. However, it is not complete, because the all-zero infinite sequence has no prefix in $\mathcal{D}$.

\section{Generalized universal coding of integers}\label{sec_lb}
In this section, we first define GUCI and then explain the rationality of the definition.
Let $\mathcal{C}=(\mathcal{D},\varphi)$ denote a VV code $\mathcal{C}$ with a proper and ASC dictionary $\mathcal{D}$ and a VV encoder $\varphi$.
A VV code $\mathcal{C}=(\mathcal{D},\varphi)$, that satisfies the prefix property, means that $\varphi(\beta)$ is not a prefix of $\varphi(\alpha)$ for any $\beta\neq\alpha\in\mathcal{D}$.
By introducing the VV codes, the definition of GUCI is as follows.
\begin{definition}\label{def:guci}
(GUCI) Let $\mathcal{C}=(\mathcal{D},\varphi)$ be a VV code that satisfies the prefix property, and it maps the non-negative integer strings $\mathbb{N}^{*}$ onto binary codewords $\{0,1\}^{*}$.
$\mathcal{C}$ is called generalized universal if there exists a constant $K_{\mathcal{C}}$ independent of $P$, for all DPD $P$ with $0<H(P)<\infty$, such that
\begin{equation}\label{e3}
\frac{R_{\mathcal{C}}}{H(P)}\leq K_{\mathcal{C}},
\end{equation}
where $R_{\mathcal{C}}$ is the coding rate of $\mathcal{C}$, $K_{\mathcal{C}}$ denotes the expansion factor of GUCI $\mathcal{C}$ and $K_{\mathcal{C}}^{*}\triangleq\inf\{K_{\mathcal{C}}\}$ denotes the minimum expansion factor of GUCI $\mathcal{C}$.
GUCI $\mathcal{C}$ is called optimal if $\mathcal{C}$ achieves the smallest minimum expansion factor $K_{\mathcal{C}}^{*}$.
\end{definition}

Next, we discuss the rationality of Definition \ref{def:guci}.
First of all, it is explained that the definition of GUCI is an extension of UCI.
Comparing inequality \eqref{e1} with inequality \eqref{e3}, 
since the denominator of the fraction on the left-hand side of the inequality removes the $\max$ function,
it is extended from this perspective. 
The numerators of the fractions on the left-hand side of the two inequalities are essentially equivalent.
Since the variable-length code is a special VV code, when the dictionary $\mathcal{D}$ of the VV code is equal to the alphabet $\mathbb{N}$, the VV code degenerates into a variable-length code.
Note that when the VV code $\mathcal{C}=(\mathcal{D},\varphi)=(\mathbb{N},\varphi)$, $\mathcal{C}$ is a variable-length code with the coding rate
\begin{equation*}
\begin{aligned}
     R_{\mathcal{C}}&=\frac{\sum_{\alpha\in \mathbb{N}}P(\alpha)|\varphi(\alpha)|}{\sum_{\alpha\in \mathbb{N}}P(\alpha)\times 1}   \\
                    &=\sum_{n=0}^{\infty}P(n)|\varphi(n)|           \\
                    &=E_{P}(L_{\varphi}).
\end{aligned}
\end{equation*}
At this time, $R_{\mathcal{C}}$ denotes the expected codeword length of $\mathcal{C}$.
Thus, $E_{P}(L_{\mathcal{C}})$ is a special $R_{\mathcal{C}}$.
Essentially, both $R_{\mathcal{C}}$ and $E_{P}(L_{\mathcal{C}})$ represent the average codeword length required for a source symbol. 
Therefore, for convenience, $R_{\mathcal{C}}$ and $E_{P}(L_{\mathcal{C}})$ can be collectively referred to as the average codeword length. 
Suppose a variable-length code $\mathcal{C}=(\mathbb{N},\varphi)$ is a class of GUCI (although such a variable-length code does not exist). Due to
\[
  \frac{E_{P}(L_{\varphi})}{\max\{1,H(P)\}}\leq\frac{R_{\mathcal{C}}}{H(P)}\leq K_{\mathcal{C}},
\]
$\mathcal{C}=(\mathbb{N},\varphi)$ is also a class of UCI.

Secondly, we prove that the expansion factor of GUCI has the same property as UCI.
In the groundbreaking paper~\cite{Elias75}, Elias proved that $E_{P}(L_{\mathcal{C}})\geq \max\{1,H(P)\}$.
Therefore, the expansion factor of UCI is greater than or equal to $1$.
Before giving the relevant theorem, we first introduce an important lemma.
\begin{lemma}~\cite{VV00} \label{lemma1}
Let $S=(P,\mathcal{A})$ denote a discrete memoryless source with entropy $H(P)<\infty$ and a countable alphabet $\mathcal{A}$.
Given a VV code $\mathcal{C}$ with a proper and ASC dictionary $\mathcal{D}$, than
\begin{equation*}
H(\mathcal{D})= H(P)\overline{l(\mathcal{D})},
\end{equation*}
where $H(\mathcal{D})=-\sum_{\alpha\in\mathcal{D}}P(\alpha)\log_{2}P(\alpha)$ denotes the entropy of $\mathcal{D}$
and $\overline{l(\mathcal{D})}=\sum_{\alpha\in\mathcal{D}}P(\alpha)|\alpha|$ denotes the average length of $\mathcal{D}$.
\end{lemma}
Lemma~\ref{lemma1} was first introduced by Nishiara \emph{et al.}~\cite{VV00}, but they did not give complete proof.
The proof for the proper and complete dictionary and the finite alphabet can be found in~\cite{FK72}.
When studying the entropy of randomly stopped sequences, 
Ekroot \emph{et al.}~\cite{LT91} gave the proof of the proper and ASC dictionary and the finite alphabet version of Lemma~\ref{lemma1}.
In addition, a similar lemma, called \emph{conservation of entropy}~\cite{S99}, is for memory sources. Therefore, we will give the first complete proof of Lemma~\ref{lemma1} in the Appendix. 
Next, we give a theorem similar to $E_{P}(L_{\mathcal{C}})\geq\max\{1,H(P)\}$ in variable-length codes.
\begin{theorem}\label{thm1}
Let $S=(P,\mathcal{A})$ denote a discrete memoryless source  with entropy $H(P)<\infty$ and a countable alphabet $\mathcal{A}$. 
Assuming that a VV code $\mathcal{C}=(\mathcal{D},\varphi)$ satisfies the prefix property, then $R_{\mathcal{C}}\geq H(P)$.
\end{theorem}
\begin{proof}
From Lemma~\ref{lemma1}, we obtain
\begin{equation*}
\begin{aligned}
     &R_{\mathcal{C}}=\frac{\sum_{\alpha\in\mathcal{D}}P(\alpha)|\varphi(\alpha)|}{\overline{l(\mathcal{D})}}\geq H(P)=\frac{H(\mathcal{D})}{\overline{l(\mathcal{D})}}   \\
\iff &\sum_{\alpha\in\mathcal{D}}P(\alpha)|\varphi(\alpha)|\geq H(\mathcal{D})=-\sum_{\alpha\in\mathcal{D}}P(\alpha)\log_{2}P(\alpha)       \\
\iff &\sum_{\alpha\in\mathcal{D}}P(\alpha)\Big(|\varphi(\alpha)|+\log_{2}P(\alpha)\Big)\geq 0     \\
\iff &\sum_{\alpha\in\mathcal{D}}P(\alpha)\log_{2}\frac{P(\alpha)}{2^{-|\varphi(\alpha)|}}\geq 0.
\end{aligned}
\end{equation*}
Below we prove that the last inequality holds.
As the codeword set $\{\varphi(\alpha)\mid \alpha\in\mathcal{D}\}$ satisfies the prefix property, we have
\[
  \sum_{\alpha\in\mathcal{D}}2^{-|\varphi(\alpha)|}\leq 1
\]
due to Kraft inequality~\cite{kraft}.
We can find the set $\{\psi(\alpha)\mid \alpha\in\mathcal{D}\}$ that satisfies
\[
  \sum_{\alpha\in\mathcal{D}}2^{-|\psi(\alpha)|}= 1
\]
and $|\varphi(\alpha)|\geq |\psi(\alpha)|$, for every $\alpha\in\mathcal{D}$.
Then,
\begin{equation*}
\begin{aligned}
\sum_{\alpha\in\mathcal{D}}P(\alpha)\log_{2}\frac{P(\alpha)}{2^{-|\varphi(\alpha)|}} & \geq \sum_{\alpha\in\mathcal{D}}P(\alpha)\log_{2}\frac{P(\alpha)}{2^{-|\psi(\alpha)|}}     \\
                                                                                     & =   D(P\parallel P_\psi)  \\
                                                                                     & \geq 0,
\end{aligned}
\end{equation*}
where the probability distribution represented by $P_{\psi}$ satisfies $P_{\psi}(\alpha)=2^{-|\psi(\alpha)|}$, for
every $\alpha\in\mathcal{D}$, and $D(P\parallel P_\psi)$ denotes relative entropy.
\end{proof}
From Theorem~\ref{thm1}, we obtain that the expansion factor of GUCI is greater than or equal to $1$.

\section{Explicit construction of GUCI}\label{sec_new}
In this section, the explicit structure of a family of GUCI is proposed.
The traditional UCI cannot satisfy inequality \eqref{e2}, as there is no constant $K_{\mathcal{C}}$ to meet inequality \eqref{e2} when $H(P)$ tends to $0$.
Thus, we pay attention to the case that $H(P)$ tends to $0$ on the construction of GUCI.
When $H(P)$ tends to $0$, $P(0)$ tends to $1$, the non-negative integer source string  will contain several consecutive $0'$s, that can be compressed by RLE. Precisely, the proposed VV code $\mathcal{C}=(\mathcal{D},\varphi)$ is the concatenation of RLE and UCI $\psi$. The encoding process is as follows.

First, the dictionary $\mathcal{D}_{RLE}$ selected by the encoder is
\[
  \mathcal{D}_{RLE}=\{ \underbrace {\l {00\cdots 0}}_i n | i\in \mathbb{N}, n\in \mathbb{N}^{+}  \}.
\]
Next, the encoder maps the variable-length string $\underbrace {\l {00\cdots 0}}_i n\in \mathcal{D}_{RLE}$ into the fixed-length string $(i+1,n)$.
Finally, the encoder maps string $(i+1,n)$ into $\psi(i+1)\psi(n)$ by UCI $\psi$.
That is, $\varphi_{\psi}(\underbrace {\l {00\cdots 0}}_i n)=\psi(i+1)\psi(n)$.

Obviously, $\mathcal{D}_{RLE}$ is proper and not complete, as the all-zero infinite sequence has no prefix in $\mathcal{D}_{RLE}$.
However, as the probability of the all-zero infinite sequence is $0$ due to $H(P)>0$, $\mathcal{D}_{RLE}$ is ASC.
We prove that the constructed VV code $\mathcal{C}=(\mathcal{D}_{RLE},\varphi_{\psi})$ is GUCI when UCI $\psi$ meets an easily reachable condition below.
First, we give two auxiliary lemmas.
\begin{lemma}\label{lemma3}
The following inequality holds.
\begin{equation*}
  -\log_{2}\Big(P(0)^{i}P(n)\Big)\geq  1+\log_{2}n+\log_{2}(i+1),
\end{equation*}
for all DPD $P$ and every $i\in \mathbb{N}^{+}$ and $n\in \mathbb{N}^{+}$.
\end{lemma}
\begin{proof}
Since $P$ is DPD, we obtain
\[
P(0)^{i}P(n)\leq\frac{P(0)^{i}\Big(1-P(0)\Big)}{n}.
\]
Let $g(x)=x^{i}(1-x)$, for $0<x<1$.
We know that $g(x)$ is strictly increasing when $x\in(0,\frac{i}{i+1})$ and $g(x)$ is strictly decreasing when $x\in(\frac{i}{i+1},1)$ via its derivative.
Thus,
\[
  P(0)^{i}P(n)\leq \frac{1}{n}\cdot g\left(\frac{i}{i+1}\right)=\frac{1}{n}\cdot\frac{1}{i+1}\cdot\left(\frac{i}{i+1}\right)^{i}.
\]
We prove that the sequence $\{a_{i}=(\frac{i}{i+1})^{i}\}_{i=1}^{\infty}$ is strictly monotonically decreasing below.
Let
\[
  b_{i}=\frac{1}{a_{i}}=\left(\frac{i+1}{i}\right)^{i}=\left(1+\frac{1}{i}\right)^{i},
\]
then $\{a_{i}\}_{i=1}^{\infty}$ strictly monotonically decreasing is equivalent to 
$\{b_{i}\}_{i=1}^{\infty}$ strictly monotonically increasing.
Due to the inequality of arithmetic and geometric means, we obtain
\begin{equation*}
\begin{aligned}
     b_{i}& = 1\cdot\underbrace {\l {\left(1+\frac{1}{i}\right)\cdots \left(1+\frac{1}{i}\right)}}_{\textstyle i}   \\
          & < \left[\frac{1+i(1+\frac{1}{i})}{i+1}\right]^{i+1}   \\
          & = \left(1+\frac{1}{i+1}\right)^{i+1}   \\
          & =b_{i+1}.
\end{aligned}
\end{equation*}
Thus,
\[
  P(0)^{i}P(n)\leq \frac{1}{n}\cdot\frac{1}{i+1}\cdot\left(\frac{i}{i+1}\right)^{i}\leq \frac{1}{n}\cdot\frac{1}{i+1}\cdot\frac{1}{2},
\]
and hence,
\[
   -\log_{2}\Big(P(0)^{i}P(n)\Big)\geq  1+\log_{2}n+\log_{2}(i+1).
\]

\end{proof}
\begin{lemma}\label{lemma2}
Given two positive numbers $a$ and $b$, then
\begin{equation} \label{eq4}
2a+b\log_{2}n+b\log_{2}(i+1)\leq -(2a+b)\log_{2}\Big(P(0)^{i}P(n)\Big),
\end{equation}
for all DPD $P$ and every $i\in \mathbb{N}$ and $n\in \mathbb{N}^{+}$.
\end{lemma}
\begin{proof}
We first consider $i=0$. In this case, inequality~\eqref{eq4} can be rewritten as
\[
 2a+b\log_{2}n\leq -(2a+b)\log_{2}P(n).
\]
As $P(0)\geq P(1)\geq \cdots \geq P(n)\geq \cdots$, then
\[
  1=\sum_{m=0}^{\infty}P(m)\geq\sum_{m=0}^{n}P(m)\geq (n+1)P(n),
\]
and hence, $-\log_{2}P(n)\geq \log_{2}(n+1)$, for $n\in \mathbb{N}^{+}$.
Thus,
\begin{equation*}
\begin{aligned}
     -(2a+b)\log_{2}P(n)& \geq(2a+b)\log_{2}(n+1)   \\
                        & =2a\log_{2}(n+1)+ b\log_{2}(n+1)   \\
                        &> 2a+ b\log_{2}n.
\end{aligned}
\end{equation*}
Then, we consider $i\geq1$.
Due to Lemma~\ref{lemma3}, we have
\begin{equation*}
\begin{aligned}
 -(2a+b)\log_{2}\Big(P(0)^{i}P(n)\Big)  &  \geq  \ (2a+b)\Big(1+\log_{2}n+\log_{2}(i+1)\Big)   \\
                                &  >  \ 2a+b\log_{2}n+b\log_{2}(i+1).
\end{aligned}
\end{equation*}
\end{proof}

Now, we give the main theorem in this section.
\begin{theorem}\label{thm2}
Let $S=(P,\mathcal{A})$ denote a discrete memoryless source with entropy $0<H(P)<\infty$ and a countable alphabet $\mathcal{A}$.
Given the VV code $\mathcal{C}=(\mathcal{D}_{RLE},\varphi_{\psi})$ satisfying
\begin{equation}\label{eq5}
L_{\psi}(n)\leq a+b\log_{2}n, \text{ for } n\in \mathbb{N}^{+},
\end{equation}
where $a$ and $b$ are two positive constants, then we have
\[
 \frac{R_{\mathcal{C}}}{H(P)}\leq 2a+b,
\]
for all DPD $P$.
\end{theorem}
\begin{proof}
From Lemma~\ref{lemma2} and inequality~\eqref{eq5}, we obtain
\begin{equation}\label{eq9}
\begin{aligned}
|\varphi_{\psi}(\underbrace {\l {00\cdots 0}}_i n)|&= L_{\psi}(n)+L_{\psi}(i+1)   \\
                                                   &\leq 2a+b\log_{2}n+b\log_{2}(i+1)    \\
                                                   &\leq -(2a+b)\log_{2}\Big(P(0)^{i}P(n)\Big),
\end{aligned}
\end{equation}
for every $i\in \mathbb{N}$ and $n\in \mathbb{N}^{+}$.
From Lemma~\ref{lemma1}, we have
\begin{equation*}
\begin{aligned}
     \frac{R_{\mathcal{C}}}{H(P)}&=\frac{\sum_{\alpha\in\mathcal{D}_{RLE}}P(\alpha)|\varphi_{\psi}(\alpha)|}{H(\mathcal{D}_{RLE})}   \\
                                 &=\frac{\sum_{i,n}P(0)^{i}P(n)\Big(L_{\psi}(n)+L_{\psi}(i+1)\Big)}{-\sum_{i,n}P(0)^{i}P(n)\log_{2}\Big(P(0)^{i}P(n)\Big)}  \\
                                 &\overset{(a)}{\leq}\frac{\sum_{i,n}P(0)^{i}P(n)\Big[-(2a+b)\log_{2}\Big(P(0)^{i}P(n)\Big)\Big]}{-\sum_{i,n}P(0)^{i}P(n)\log_{2}\Big(P(0)^{i}P(n)\Big)}  \\
                                 &= 2a+b, \\
\end{aligned}
\end{equation*}
where $(a)$ is due to inequality~\eqref{eq9}.
\end{proof}
\begin{remark}
A variable-length code $\psi$ satisfying inequality~\eqref{eq5} is a sufficient condition for $\psi$ to be UCI~\cite{L81}.
To the best of our knowledge, all UCI codes currently proposed meet inequality~\eqref{eq5}.
Therefore, when we construct a GUCI $\mathcal{C}=(\mathcal{D}_{RLE},\varphi_{\psi})$, we can choose any known UCI code.
\end{remark}

\section{The tighter upper bound of $K_{\mathcal{C}}^{*}$ for optimal GUCI}\label{factor}
Based on UCIs, the prior section provides a family of GUCIs. 
This section explores the expansion factors of some specific GUCIs and obtains the tighter upper bound of $K_{\mathcal{C}}^{*}$ for optimal GUCI.
For any UCI $\mathcal{C}$, its expansion factor $K_{\mathcal{C}}$ is greater than or equal to $2$~\cite{YL21}. The best known UCI to date is the $\iota$ code~\cite{YL22} with $K_{\mathcal{\iota}}=2.5$.
Therefore, the optimal UCI is in the range $2\leq K_{\mathcal{C}}^{*}\leq 2.5$.
Theorem~\ref{thm1} shows that $K_{\mathcal{C}}^{*}$ of the optimal GUCI is greater than or equal to 1.
This section investigates the tighter upper bounds of $K_{\mathcal{C}}^{*}$ for optimal GUCI.

When constructing a VV code $\mathcal{C}=(\mathcal{D}_{RLE},\varphi_{\psi})$, we select Elias $\gamma$ code as the UCI $\psi$.
From Theorem~\ref{thm2} and $L_{\gamma}(n)=1+2\lfloor \log_{2}n\rfloor \leq 1+2\log_{2}n$, we obtain $K_{\mathcal{C}}=4$. It is showed~\cite{ITW} that $\mathcal{C}=(\mathcal{D}_{RLE},\varphi_{\gamma})$ can achieve $K_{\mathcal{C}}=\frac{6}{\log_{2}5}\approx2.584$.
However, this result is not tight, and we will show that the VV code $\mathcal{C}=(\mathcal{D}_{RLE},\varphi_{\gamma})$ can achieve $K_{\mathcal{C}}=2$. First, we give two auxiliary lemmas.

\begin{lemma}\label{lemma4}
For all DPD $P$ defined on $\mathbb{N}$ and all $m\in\mathbb{N}^{+}$, we obtain
\begin{enumerate}
\item $\sum_{j=1}^{m}P(j)\leq\frac{m}{m+1}$;
\item $\prod_{j=1}^{m}P(j)\leq\left(\frac{1}{m+1}\right)^{m}$;
\item Let $A_{m}\triangleq 2^{m}\times m! \times\left(\frac{1}{m+1}\right)^{m}$, then $A_{m}\leq1$;
\item Let $B_{m}\triangleq \sum_{j=1}^{m}\Big(1+\log_{2}j+\log_{2}P(j)\Big)$, then $B_{m}\leq0$.
\end{enumerate}
\end{lemma}
\begin{proof}
\begin{enumerate}
\item We prove that $\sum_{j=1}^{m}P(j)\leq\frac{m}{m+1}$ by contradiction.
Suppose there exists a DPD $P_{0}$ defined on $\mathbb{N}$ such that $\sum_{j=1}^{m}P_{0}(j)>\frac{m}{m+1}$.
Thus,
\[
 P_{0}(0)\leq1-\sum_{j=1}^{m}P_{0}(j)<\frac{1}{m+1},
\]
and hence,
\[
 \frac{m}{m+1} >mP_{0}(0) \geq\sum_{j=1}^{m}P_{0}(j) >\frac{m}{m+1},
\]
that is a contradiction. Thus, the assumption is not true.
\item
Dut to the inequality of arithmetic and geometric means, we obtain
\[
\prod_{j=1}^{m}P(j)\leq \left(\frac{\sum_{j=1}^{m}P(j)}{m}\right)^{m}\leq\left(\frac{1}{m+1}\right)^{m}.
\]
\item We prove that $A_{m}\leq1$ by mathematical induction.
When $m=1$, then $A_{1}=2\times1\times\frac{1}{2}\leq1$.
Suppose $A_{m}\leq1$ holds when $m=n$.
When $m=n+1$, we have
\begin{equation*}
\begin{aligned}
     A_{n+1}&=A_{n}\times 2(n+1)\left(\frac{1}{n+2}\right)^{n+1}\div\left(\frac{1}{n+1}\right)^{n}   \\
            &=2A_{n}\times \left(\frac{n+1}{n+2}\right)^{n+1} \\
             &\overset{(a)}{\leq}2A_{n}\times \left(\frac{2}{3}\right)^{2}  \\
             &=\frac{8}{9}A_{n}   \\
             &<1, \\
\end{aligned}
\end{equation*}
where $(a)$ is from that fact that the sequence $\{a_{i}=(\frac{i}{i+1})^{i}\}_{i=1}^{\infty}$ is strictly monotonically decreasing.
\item
From above results, we obtain
\begin{equation*}
\begin{aligned}
     B_{m}&=\sum_{j=1}^{m}\log_{2}\Big(2\times j\times P(j)\Big)   \\
          &=\log_{2}\left(2^{m}\times m!\times \prod_{j=1}^{m}P(j) \right) \\
          &\leq\log_{2}\left(2^{m}\times m!\times \left(\frac{1}{m+1}\right)^{m} \right)  \\
          &=\log_{2}A_{m}   \\
          &\leq0.\\
\end{aligned}
\end{equation*}
\end{enumerate}
\end{proof}
\begin{lemma}\label{lemma5}
For all DPD $P$ defined on $\mathbb{N}$ and all $m\in\mathbb{N}^{+}$, we define
\[
 S_{m}\triangleq \sum_{j=1}^{m}P(j)\Big(1+\log_{2}j+\log_{2}P(j)\Big).
\]
Then, $S_{m}\leq0$, for all $m\in\mathbb{N}^{+}$. Then we obtain
\[
\sum_{j=1}^{\infty}P(j)\Big(1+\log_{2}j+\log_{2}P(j)\Big)\leq0.
\]
\end{lemma}
\begin{proof}
When $m=1$, we have $S_{1}=P(1)B_{1}\leq0$. When $m\geq2$, we obtain
\begin{equation*}
\begin{aligned}
     S_{m}=&\ P(1)B_{1}+\sum_{j=2}^{m}P(j)(B_{j}-B_{j-1}) \\
          =&\ \Big(P(1)-P(2)\Big)B_{1}+P(2)B_{2}+\sum_{j=3}^{m}P(j)(B_{j}-B_{j-1}) \\
          \leq &\  P(2)B_{2}+\sum_{j=3}^{m}P(j)(B_{j}-B_{j-1}) \\
           \vdots &  \\
          \leq & \ P(m-1)B_{m-1}+P(m)(B_{m}-B_{m-1}) \\
          \leq & \ P(m)B_{m}    \\
          \leq & \ 0.
\end{aligned}
\end{equation*}
Thus, we have
\[
\sum_{j=1}^{\infty}P(j)\Big(1+\log_{2}j+\log_{2}P(j)\Big)=\lim\limits_{m\to +\infty}S_{m}\leq \lim\limits_{m\to +\infty}0 = 0.
\]
\end{proof}
Now, we give the main result of this section.
\begin{theorem}\label{thm5}
Let $S=(P,\mathcal{A})$ denote a discrete memoryless source with entropy $0<H(P)<\infty$ and a countable alphabet $\mathcal{A}$.
Given a VV code $\mathcal{C}=(\mathcal{D}_{RLE},\varphi_{\psi})$ satisfying
\begin{equation} \label{eq25}
L_{\psi}(n)\leq b+2b\log_{2}n, \text{ for } n\in \mathbb{N}^{+},
\end{equation}
where $b$ is a positive constant, then
\[
 \frac{R_{\mathcal{C}}}{H(P)}\leq 2b,
\]
for all DPD $P$.
\end{theorem}
\begin{proof}
From Lemma~\ref{lemma1} and inequality~\eqref{eq25}, we obtain
\begin{equation*}
\begin{aligned}
     \frac{R_{\mathcal{C}}}{H(P)}&=\frac{\sum_{\alpha\in\mathcal{D}_{RLE}}P(\alpha)|\varphi_{\psi}(\alpha)|}{H(P)\overline{l(\mathcal{D}_{RLE})}}   \\
                                 &=\frac{\sum_{i=0}^{\infty}\sum_{n=1}^{\infty}P(0)^{i}P(n)\Big(L_{\psi}(i+1)+L_{\psi}(n)\Big)}{H(\mathcal{D}_{RLE})}   \\
                                 &\leq 2b\cdot \frac{\sum_{i=0}^{\infty}\sum_{n=1}^{\infty}P(0)^{i}P(n)\Big(1+\log_{2}n+\log_{2}(i+1)\Big)}{H(\mathcal{D}_{RLE})} . \\
\end{aligned}
\end{equation*}
Therefore, proving $\frac{R_{\mathcal{C}}}{H(P)}\leq 2b$ is equivalent to show that
\begin{equation} \label{eq17}
 \sum_{i=0}^{\infty}\sum_{n=1}^{\infty}P(0)^{i}P(n)\Big(1+\log_{2}n+\log_{2}(i+1)\Big)\leq H(\mathcal{D}_{RLE}).
\end{equation}
When $i\geq1$, from Lemma~\ref{lemma3}, we have
\[
  1+\log_{2}n+\log_{2}(i+1)\leq -\log_{2}\Big(P(0)^{i}P(n)\Big).
\]
Thus, we obtain
\begin{equation} \label{eq15}
 \sum_{i=1}^{\infty}\sum_{n=1}^{\infty}P(0)^{i}P(n)\Big(1+\log_{2}n+\log_{2}(i+1)\Big)\leq -\sum_{i=1}^{\infty}\sum_{n=1}^{\infty}P(0)^{i}P(n)\log_{2}\Big(P(0)^{i}P(n)\Big).
\end{equation}
When $i=0$, from Lemma~\ref{lemma5}, we have
\begin{equation} \label{eq16}
  \sum_{n=1}^{\infty}P(n)\left(1+\log_{2}n\right)\leq -\sum_{n=1}^{\infty}P(n)\log_{2}P(n).
\end{equation}
From inequality \eqref{eq15} and inequality \eqref{eq16}, inequality \eqref{eq17} holds.
\end{proof}
\begin{remark}
As $1\leq L_{\psi}(1)\leq b$, the minimum of $b$ in Theorem~\ref{thm5} is $1$.
From Theorem~\ref{thm5} and $L_{\gamma}(n)\leq 1+2\log_{2}n$, we know that $\mathcal{C}=(\mathcal{D}_{RLE},\varphi_{\gamma})$ can achieve $K_{\mathcal{C}}=2$.
Thus, Elias $\gamma$ code achieves the best case of Theorem~\ref{thm5}.
\end{remark}
Next, we discuss $K_C$ for GUCIs constructed using other classical UCIs.
First, the following lemma is proved.

\begin{lemma}\label{lemma6}
For all $n\in\mathbb{N}^{+}$,
\begin{enumerate}
\item the length function of $\delta$ code satisfies $L_{\delta}(n)\leq \frac{4}{3}+\frac{8}{3}\log_{2}n$;
\item the length function of $\eta$ code satisfies $L_{\eta}(n)\leq \frac{6}{1+2\log_{2}5}+\frac{12}{1+2\log_{2}5}\log_{2}n$;
\item the length function of $\theta$ code satisfies $L_{\theta}(n)\leq \frac{4}{3}+\frac{8}{3}\log_{2}n$;
\item the length function of $\iota$ code satisfies  $L_{\iota}(n)\leq \frac{4}{3}+\frac{8}{3}\log_{2}n$;
\item the length function of $\omega$ code satisfies $L_{\omega}(n)\leq \frac{11}{9}+\frac{22}{9}\log_{2}n$.
\end{enumerate}
\end{lemma}
\begin{proof}
\begin{enumerate}
\item Obviously, the inequality $\lfloor\log_{2}(1+x)\rfloor\leq\frac{1}{6}+\frac{5}{6}x$ holds, for all $x\in\mathbb{N}$.
Thus, we obtain
\begin{equation*}
\begin{aligned}
     L_{\delta}(n)&=1+\lfloor\log_{2}n\rfloor+2\lfloor\log_{2}(1+\lfloor\log_{2}n\rfloor)\rfloor  \\
                  &\leq1+\lfloor\log_{2}n\rfloor+2\left(\frac{1}{6}+\frac{5}{6}\lfloor\log_{2}n\rfloor\right)   \\
                  &\leq \frac{4}{3}+\frac{8}{3}\log_{2}n.
\end{aligned}
\end{equation*}
\item Let $f(n)\triangleq \frac{6}{1+2\log_{2}5}+\frac{12}{1+2\log_{2}5}\log_{2}n$.
We directly verify $L_{\eta}(n)\leq f(n)$, for $n<16$.
When $n\geq16$, we have
\begin{equation*}
\begin{aligned}
     L_{\eta}(n)  &=3+\lfloor\log_{2}(n-1)\rfloor+\lfloor\frac{\lfloor\log_{2}(n-1)\rfloor}{2}\rfloor  \\
                  &\leq 3+\frac{3}{2}\lfloor\log_{2}n\rfloor   \\
                  &\leq 1+2\lfloor\log_{2}n\rfloor   \\
                  &<f(n).
\end{aligned}
\end{equation*}
\item
Obviously, the inequality $\frac{5}{3}+\frac{3}{2}\lfloor\log_{2}x\rfloor\leq\frac{5}{3}x$ holds, for all $x\in\mathbb{N}^{+}$.
Thus, we obtain $ L_{\theta}(1)=1<\frac{4}{3}$ and
\begin{equation*}
\begin{aligned}
     L_{\theta}(n)&=3+\lfloor\log_{2}n\rfloor+\lfloor\log_{2}\lfloor\log_{2}n\rfloor\rfloor+\lfloor\frac{\lfloor\log_{2}\lfloor\log_{2}n\rfloor\rfloor}{2}\rfloor  \\
                  &\leq 3+\lfloor\log_{2}n\rfloor+\frac{3}{2}\lfloor\log_{2}\lfloor\log_{2}n\rfloor\rfloor \\
                  & = \frac{4}{3}+\lfloor\log_{2}n\rfloor+\left(\frac{5}{3}+\frac{3}{2}\lfloor\log_{2}\lfloor\log_{2}n\rfloor\rfloor\right)    \\
                  & \leq \frac{4}{3}+\frac{8}{3}\log_{2}n,
\end{aligned}
\end{equation*}
for $n\geq2$.
\item We obtain $ L_{\iota}(1)=1<\frac{4}{3}$ and
\begin{equation*}
\begin{aligned}
     L_{\theta}(n)&=2+\lfloor\log_{2}n\rfloor+\lfloor\frac{1+\lfloor\log_{2}n\rfloor}{2}\rfloor  \\
                  &\leq \frac{5}{2}+\frac{3}{2}\lfloor\log_{2}n\rfloor \\
                  & =\frac{4}{3}+\frac{8}{3}\lfloor\log_{2}n\rfloor+\frac{7}{6}\left( 1-\lfloor\log_{2}n\rfloor \right) \\
                  & \leq \frac{4}{3}+\frac{8}{3}\log_{2}n,
\end{aligned}
\end{equation*}
for $n\geq2$.
\item
We directly verify $L_{\omega}(n)\leq \frac{11}{9}+\frac{22}{9}\log_{2}n$, for $n<16$.
When $n\geq16$, we have
\begin{equation*}
\begin{aligned}
     L_{\omega}(n) &\leq3+2\lfloor\log_{2}n\rfloor  \\
                   &=\frac{11}{9}+\frac{22}{9}\lfloor\log_{2}n\rfloor+\frac{4}{9}\left( 4-\lfloor\log_{2}n\rfloor \right)  \\
                  &\leq \frac{11}{9}+\frac{22}{9}\log_{2}n.
\end{aligned}
\end{equation*}
\end{enumerate}
\end{proof}
\begin{table}[t]
\centering
\caption{The expansion factors that can be achieved for VV code $\mathcal{C}=(\mathcal{D}_{RLE},\varphi_{\psi})$}\label{tab1}
\begin{tabular}{|c|c|}
\hline
UCI $\psi$  &   expansion factor $K_{\mathcal{C}}$   \\
\hline
$\gamma$ code  &    $2$            \\
$\eta$ code    &    $\frac{12}{1+2\log_{2}5}\approx2.13$       \\
$\omega$ code  &    $\frac{22}{9}\approx2.44$         \\
$\delta$ code, $\theta$ code and $\iota$ code  &    $\frac{8}{3}\approx2.67$    \\
\hline
\end{tabular}
\end{table}
From Theorem~\ref{thm5} and Lemma~\ref{lemma6}, Table~\ref{tab1} lists the expansion factors of GUCIs when choosing various UCIs.
Note that based on previous proofs, we are aware that the range of the minimum expansion factor of the optimal GUCI is $1\leq K_{\mathcal{C}}^{*}\leq2$.

\section{Comparison of the average codeword lengths of UCI $\psi$ and GUCI $\mathcal{C}=(\mathcal{D}_{RLE},\varphi_{\psi})$}\label{length}

In this section, we compare the expected codeword length $E_{P}(L_{\mathcal{\psi}})$ of UCI $\psi$ 
and the coding rate $R_{\mathcal{C}}$ of GUCI $\mathcal{C}=(\mathcal{D}_{RLE},\varphi_{\psi})$. 
Intuitively, when Shannon entropy $H(P)$ is small or $P(0)$ is large, $E_{P}(L_{\mathcal{\psi}})>R_{\mathcal{C}}$;
when Shannon entropy $H(P)$ is large or $P(0)$ is small, $E_{P}(L_{\mathcal{\psi}})<R_{\mathcal{C}}$. 
The following two conclusions can be drawn from the research in this section. 
One is that when $P(0)$ is relatively large, $E_{P}(L_{\mathcal{\psi}})>R_{\mathcal{C}}$; that is, 
a sufficient condition for the average codeword length of GUCI to be shorter than UCI is obtained.
The second is that when Shannon entropy $H(P)$ is very large or $P(0)$ is not large, 
there are still cases where $E_{P}(L_{\mathcal{\psi}})>R_{\mathcal{C}}$. 
A detailed discussion is given below.

To begin with, we recall a definition. If a class of UCI $\psi$ satisfies
\begin{equation}\label{eq21}
 L_{\psi}(m)\leq L_{\psi}(m+1), \hbox{  for } m\in\mathbb{N},
\end{equation}
then $\psi$ is termed \emph{minimal}~\cite{Elias75}.
For all DPD $P$, $E_{P}(L_{\mathcal{C}})$ can be minimized when inequality~\eqref{eq21} is satisfied, hence the definition is natural.

The $E_{P}(L_{\mathcal{\psi}})$ and $R_{\mathcal{C}}$ are defined as
\[
 E_{P}(L_{\mathcal{\psi}})=\sum_{n=0}^{\infty}P(n)L_{\psi}(n+1),
\]
\begin{equation*}
 R_{\mathcal{C}}=\frac{\sum_{\alpha\in\mathcal{D}_{RLE}}P(\alpha)|\varphi_{\psi}(\alpha)|}{\overline{l(\mathcal{D}_{RLE})}}.
\end{equation*}
Furthermore, we obtain
\begin{equation*}
\begin{aligned}
  \overline{l(\mathcal{D}_{RLE})}&= \sum_{i=0}^\infty\sum_{n=1}^{\infty}P(0)^{i}P(n)(i+1)        \\
                                 &=\sum_{n=1}^{\infty}P(n)\left(\sum_{i=0}^\infty P(0)^{i}+\sum_{i=0}^\infty iP(0)^{i}\right)\\
                                 &=\Big(1-P(0)\Big)\left(\frac{1}{1-P(0)}+\frac{P(0)}{\Big(1-P(0)\Big)^{2}}\right)      \\
                                 &=\frac{1}{1-P(0)},
\end{aligned}
\end{equation*}
and
\begin{equation*}
\begin{aligned}
\sum_{\alpha\in\mathcal{D}_{RLE}}P(\alpha)|\varphi_{\psi}(\alpha)|   &=\sum_{i=0}^\infty\sum_{n=1}^{\infty}P(0)^{i}P(n)\Big(L_{\psi}(i+1)+L_{\psi}(n)\Big)        \\
  &=\sum_{n=1}^{\infty}P(n)\sum_{i=0}^\infty P(0)^{i}L_{\psi}(i+1)+\sum_{i=0}^\infty P(0)^{i}\sum_{n=1}^{\infty}P(n)L_{\psi}(n)    \\
  &=\Big(1-P(0)\Big)\sum_{i=0}^\infty P(0)^{i}L_{\psi}(i+1)+\frac{\sum_{n=1}^{\infty}P(n)L_{\psi}(n)}{1-P(0)}.
\end{aligned}
\end{equation*}
Thus, we have
\begin{equation*}
 R_{\mathcal{C}}=\Big(1-P(0)\Big)^{2}\sum_{i=0}^\infty P(0)^{i}L_{\psi}(i+1)+\sum_{n=1}^{\infty}P(n)L_{\psi}(n).
\end{equation*}

Let
\begin{equation}\label{eq23}
\begin{aligned}
 \Delta &\triangleq R_{\mathcal{C}}-E_{P}(L_{\mathcal{\psi}})        \\
        &=\Big(1-P(0)\Big)^{2}\sum_{i=0}^\infty P(0)^{i}L_{\psi}(i+1)-P(0)L_{\psi}(1)-\sum_{n=1}^{\infty}P(n)\Big(L_{\psi}(n+1)-L_{\psi}(n)\Big)   \\
        &=\Big(1-P(0)\Big)^{2}\sum_{i=0}^\infty P(0)^{i}L_{\psi}(i+1)-P(0)L_{\psi}(1)-\sum_{n=1}^{\infty}P(n)\Delta_{\psi}(n)  ,
\end{aligned}
\end{equation}
where $\Delta_{\psi}(n)\triangleq L_{\psi}(n+1)-L_{\psi}(n)$ is the jump value of $\psi$ at $n$. 
$\Delta$ is a function of probability distribution $P$ and length function $L_{\psi}(\cdot)$.

When analyzing $\Delta$ without imposing restrictions on $\psi$, it is impossible to get the size of the relationship between $\Delta$ and $0$.
Then, we restrict $\psi$ with reasonable conditions and get the conclusion that $\Delta<0$ when $P(0)$ is relatively large.
The main theorem is proposed.
\begin{theorem}\label{thm6}
When constructing a VV code $\mathcal{C}=(\mathcal{D}_{RLE},\varphi_{\psi})$, the UCI $\psi$ is minimal and its length function satisfies
\begin{equation*}
L_{\psi}(n)\leq a+b\lfloor\log_{2}n\rfloor, \text{ for } 2\leq n\in \mathbb{N}^{+},
\end{equation*}
where $a$ and $b$ are two positive constants.
If there exists $t\in(0,1)$, such that
\[
 L_{\psi}(1)\left(t+\frac{1}{t}-3\right)+a\left(1-t\right)+b\left(1-t\right)\left(1+t^{2}+\frac{t^{6}}{1-t^{8}}\right)\leq 0.
\]
Then, $\Delta<0$ when $P(0)\geq t$; that is, $R_{\mathcal{C}}<E_{P}(L_{\mathcal{\psi}})$ when $P(0)\geq t$.
\end{theorem}
\begin{proof}
First, we do some calculations. We obtain
\begin{equation}\label{eq24}
\begin{aligned}
 \sum_{i=1}^\infty P(0)^{i}\lfloor\log_{2}(i+1)\rfloor &=\sum_{n=1}^\infty n\left( \sum_{j=2^{n}-1}^{2^{n+1}-2}P(0)^{j}\right)      \\
                                                       &=\sum_{n=1}^\infty n\cdot\frac{P(0)^{2^{n}-1}-P(0)^{2^{n+1}-1}}{1-P(0)}        \\
                                                       &=\frac{1}{1-P(0)}\sum_{n=1}^{\infty}P(0)^{2^{n}-1},
\end{aligned}
\end{equation}
and
\begin{equation*}
\begin{aligned}
 \sum_{i=0}^\infty P(0)^{i}L_{\psi}(i+1) &= L_{\psi}(1)+ \sum_{i=1}^\infty P(0)^{i}L_{\psi}(i+1)    \\
                                        &\leq  L_{\psi}(1)+ \frac{aP(0)}{1-P(0)}+b\sum_{i=1}^\infty P(0)^{i}\lfloor\log_{2}(i+1)\rfloor       \\
                                        &= L_{\psi}(1)+ \frac{aP(0)}{1-P(0)}+\frac{b}{1-P(0)}\sum_{n=1}^{\infty}P(0)^{2^{n}-1}   \\
                                        &< L_{\psi}(1)+ \frac{aP(0)}{1-P(0)}+\frac{bP(0)}{1-P(0)}\left(1+P(0)^{2}+\sum_{n=0}^{\infty}P(0)^{6+8n} \right)  \\
                                        &= L_{\psi}(1)+ \frac{aP(0)}{1-P(0)}+\frac{bP(0)}{1-P(0)}\left(1+P(0)^{2}+\frac{P(0)^{6}}{1-P(0)^{8}} \right).
\end{aligned}
\end{equation*}
Next, we have
\begin{equation*}
\begin{aligned}
 \Delta &\overset{(c)}{\leq} \Big(1-P(0)\Big)^{2}\sum_{i=0}^\infty P(0)^{i}L_{\psi}(i+1)-P(0)L_{\psi}(1)        \\
        &< \Big(1-P(0)\Big)^{2}\left[L_{\psi}(1)+ \frac{aP(0)}{1-P(0)}+\frac{bP(0)}{1-P(0)}\left(1+P(0)^{2}+\frac{P(0)^{6}}{1-P(0)^{8}} \right) \right]-P(0)L_{\psi}(1)   \\
        &= P(0)\left[L_{\psi}(1)\left(P(0)+\frac{1}{P(0)}-3\right)+a\Big(1-P(0)\Big)+b\Big(1-P(0)\Big)\left(1+P(0)^{2}+\frac{P(0)^{6}}{1-P(0)^{8}}\right)\right]   \\
        &\overset{(d)}{\leq} P(0) \left[L_{\psi}(1)\left(t+\frac{1}{t}-3\right)+a\left(1-t\right)+b\left(1-t\right)\left(1+t^{2}+\frac{t^{6}}{1-t^{8}}\right)\right]   \\
        &\leq 0,
\end{aligned}
\end{equation*}
where $(c)$ is because $\psi$ is minimal, $(d)$ follows the monotonic decrease in $g_{1}(x)=x+\frac{1}{x}-3$, $g_{2}(x)=1-x$ and $g_{3}=(1-x)(1+x^{2}+\frac{x^{6}}{1-x^{8}})$ over interval $(0,1)$.
\end{proof}

To give several instances, we apply some UCIs to Theorem~\ref{thm6}. In the first example, the corresponding parameters of Elias $\gamma$ code are $L_{\gamma}(1)=1$, $a=1$ and $b=2$.
Let $h(x)\triangleq(x+\frac{1}{x}-3)+(1-x)+2(1-x)(1+x^{2}+\frac{x^{6}}{1-x^{8}})$.
We know that $h(0.81)<0$ by calculation.
Thus, when $P(0)\geq0.81$, the coding rate $R_{\mathcal{C}}$ of $\mathcal{C}=(\mathcal{D}_{RLE},\varphi_{\gamma})$ is less than the expected codeword length $E_{P}(L_{\gamma})$ of Elias $\gamma$ code.
In another example, the corresponding parameters of $\iota$ code are $L_{\gamma}(1)=1$, $a=2.5$ and $b=1.5$.
We can know that when $P(0)\geq0.83$, $R_{\mathcal{C}}$ of $\mathcal{C}=(\mathcal{D}_{RLE},\varphi_{\iota})$ is less than $E_{P}(L_{\iota})$ of $\iota$ code.

Note that $R_{\mathcal{C}}$ is less than $E_{P}(L_\psi)$ not only when the entropy is small. 
In other words, $P(0)$ is relatively large, which does not mean that entropy is small.
For example, we consider the probability distribution
\[
 P_{1}=\left(P_{1}(0)=0.9,P_{1}(1)=P_{1}(2)=\cdots P_{1}(n)=\frac{1}{10n}  \right).
\]
Due to Theorem~\ref{thm6}, we obtain $R_{\mathcal{C}}<E_{P_1}(L_{\gamma})$.
However, taking the limit $n\rightarrow +\infty$, the entropy $H(P_{1})=0.1\log_{2}(10n)-0.9\log_{2}0.9$ tends to infinity.
This tells us that when the entropy is large, $R_{\mathcal{C}}$ is still less than $E_{P}(L_{\gamma})$.
But if $P(0)$ is relatively large, a long string of zeros is prone to appear.
Knowing from the structure of $\mathcal{C}=(\mathcal{D}_{RLE},\varphi_{\psi})$, it is reasonable that $R_{\mathcal{C}}$ is less than $E_{P}(L_{\psi})$ at this time.

Finally, we explore the situation when $P(0)$ is not large.
This part needs to be analyzed with a specific UCI. Considering that Elias $\gamma$ code performs best in terms of the expansion factor,
we use Elias $\gamma$ code for analysis.
Due to $L_{\gamma}(n)=1+2\lfloor\log_{2}n\rfloor$ and equation~\eqref{eq24}, equation~\eqref{eq23} can be rewritten as
\begin{equation*}
\begin{aligned}
 \Delta &=\Big(1-P(0)\Big)^{2}\sum_{i=0}^\infty P(0)^{i}\Big(1+2\lfloor\log_{2}(i+1)\rfloor\Big)-P(0)-\sum_{n=1}^{\infty}P(n)\Delta_{\gamma}(n)    \\
        &=1-2P(0)+2\Big(1-P(0)\Big)\sum_{n=1}^{\infty}P(0)^{2^{n}-1}-2\sum_{t=1}^{\infty}P(2^{t}-1).    \\
\end{aligned}
\end{equation*}
Considering the probability distribution
\[
 P_{2}=\left(P_{2}(0)=P_{2}(1)=P_{2}(2)=P_{2}(3)=0.24,P_{2}(4)=P_{2}(5)=\cdots P_{2}(n+3)=\frac{1}{25n}  \right),
\]
we obtain
\begin{equation*}
\begin{aligned}
 \Delta &< 1-2P_{2}(0)+2P_{2}(0)\Big(1-P_{2}(0)\Big)\left(1+P_{2}(0)^{2}+\frac{P_{2}(0)^{6}}{1-P_{2}(0)^{8}}\right)-2P_{2}(1)-2P_{2}(3)   \\
        &= 0.52+0.3648\times\left(1+0.0576+\frac{0.24^{6}}{1-0.24^{8}}\right)-0.96  \\
        &\approx -0.054.    \\
\end{aligned}
\end{equation*}
Taking the limit $n\rightarrow +\infty$, the entropy $H(P_{2})$ tends to infinity.
Therefore, when $P(0)$ is not large, it is still possible that $R_{\mathcal{C}}$ is less than $E_{P}(L_{\psi})$.
Note that from the calculation, it can be seen that the main reason for $\Delta<0$ in this example is the displacement term $-\sum_{n=1}^{\infty}P(n)\Delta_{\psi}(n)$.

In summary, Theorem~\ref{thm6} shows that when $P(0)$ is relatively large, $R_{\mathcal{C}}$ must be less than $E_{P}(L_{\psi})$.
When $P(0)$ is not large, it is difficult to judge whether $\Delta$ is positive or negative.
When the entropy is very large or $P(0)$ is not large, it is still possible that $\Delta$ is less than $0$.

\section{Asymptotically optimal GUCI}\label{sec_dis}
In this section, the asymptotically optimal GUCI is discussed.
First, the formal definition of asymptotically optimal GUCI is given as follows.
\begin{definition}(asymptotically optimal GUCI)
$\mathcal{C}$ is said to be asymptotically optimal GUCI, if $\mathcal{C}$ is a class of GUCI and
there exists a function $T_{\mathcal{C}}(\cdot)$ such that
\begin{equation}\label{eq7}
\frac{R_{\mathcal{C}}}{H(P)}\leq T_{\mathcal{C}}(H(P)),
\end{equation}
for all DPD $P$ with $0<H(P)<\infty$ and
\begin{equation*}
\lim\limits_{H(P)\to +\infty}T_{\mathcal{C}}(H(P))=1.
\end{equation*}
\end{definition}
Then, we give an important property about asymptotically optimal GUCI.
\begin{theorem}\label{thm3}
Let $S=(P,\mathcal{A})$ denote a discrete memoryless source with entropy $0<H(P)<\infty$ and a countable alphabet $\mathcal{A}$. Let $\mathcal{C}=(\mathcal{D}_{RLE},\varphi_{\psi})$ denote the VV code meeting inequality~\eqref{eq5} and UCI $\psi$ is minimal.
If there exists a function $R_{\psi}(\cdot)$ satisfying inequality~\eqref{eq6} and
\begin{equation*}
\lim\limits_{H(P)\to +\infty}R_{\psi}(H(P))= c,
\end{equation*}
where $c$ is constant. Then, there exists a function $T_{\mathcal{C}}(\cdot)$ satisfying inequality~\eqref{eq7} and
\begin{equation*}
\lim\limits_{H(P)\to +\infty}T_{\mathcal{C}}(H(P))=c.
\end{equation*}
\end{theorem}
\begin{proof}
From equation~\eqref{eq23}, we have
\[
 \frac{R_{\mathcal{C}}}{H(P)}=\frac{\Delta+E_{P}(L_{\psi})}{H(P)},
\]
where
\[
  \Delta = \Big(1-P(0)\Big)^{2}\sum_{i=0}^\infty P(0)^{i}L_{\psi}(i+1)-P(0)L_{\psi}(1)-\sum_{n=1}^{\infty}P(n)\Delta_{\psi}(n).
\]
From inequality~\eqref{eq5}, we know that there exists an integer $n_{0}$ such that
\begin{equation}\label{eq8}
L_{\psi}(n)\leq n, \text{ for } n_{0}\leq n\in \mathbb{N}^{+}.
\end{equation}
From inequality~\eqref{eq8} and $\Delta_{\psi}(n)\geq0$, for all $n\in\mathbb{N}$, we obtain
\begin{equation*}
\begin{aligned}
 \Delta & < \Big(1-P(0)\Big)^{2}\sum_{i=0}^\infty P(0)^{i}L_{\psi}(i+1)   \\
        &< \sum_{i=0}^{n_{0}-1}L_{\psi}(i+1)+\Big(1-P(0)\Big)^{2}\sum_{i=n_{0}}^\infty P(0)^{i}(i+1)    \\
        &= \sum_{i=0}^{n_{0}-1}L_{\psi}(i+1)+\Big(n_{0}+1-n_{0}P(0)\Big)P(0)^{n_{0}}    \\
        & \overset{(a)}{\leq} \sum_{i=0}^{n_{0}-1}L_{\psi}(i+1)+1,
\end{aligned}
\end{equation*}
where $(a)$ is because $f(x)=(n_{0}+1-n_{0}x)x^{n_{0}}$ is strictly monotonically increasing over
the interval $(0,1)$ by calculating the derivative.
Further, when $H(P)\geq1$, we obtain
\begin{equation*}
\begin{aligned}
 \frac{R_{\mathcal{C}}}{H(P)} & = \frac{\Delta+E_{P}(L_{\psi})}{H(P)}   \\
                              & < \frac{\sum_{i=0}^{n_{0}-1}L_{\psi}(i+1)+1}{H(P)}+ R_{\psi}(H(P)).
\end{aligned}
\end{equation*}
When $H(P)<1$, we have $\frac{R_{\mathcal{C}}}{H(P)}\leq 2a+b$ due to Theorem~\ref{thm2}.
We define
\[
T_{\mathcal{C}}(H(P))\triangleq\left\{\begin{array}{ll}
2a+b,                         &\text{if } H(P)<1,\\
V(H(P)),       &\text{if } H(P)\geq1,\\
\end{array}\right.
\]
where $V(H(P))\triangleq \frac{\sum_{i=0}^{n_{0}-1}L_{\psi}(i+1)+1}{H(P)}+ R_{\psi}(H(P))$.
And hence, we obtain $\frac{R_{\mathcal{C}}}{H(P)}\leq T_{\mathcal{C}}(H(P))$ and
\begin{equation*}
\begin{aligned}
   &\lim\limits_{H(P)\to +\infty}T_{\mathcal{C}}(H(P))  \\
  =& \lim\limits_{H(P)\to +\infty}\frac{\sum_{i=0}^{n_{0}-1}L_{\psi}(i+1)+1}{H(P)}+ \lim\limits_{H(P)\to +\infty} R_{\psi}(H(P))    \\
  =& \ c.
\end{aligned}
\end{equation*}
\end{proof}
Finally, we give the theorem of the relationship between the asymptotically optimal UCI and the asymptotically optimal GUCI.
\begin{theorem}\label{thm4}
For any discrete memoryless source $S=(P,\mathcal{A})$ with entropy $0<H(P)<\infty$ and a countable alphabet $\mathcal{A}$,
the VV code $\mathcal{C}=(\mathcal{D}_{RLE},\varphi_{\psi})$ meets inequality~\eqref{eq5}
and UCI $\psi$ is minimal and asymptotically optimal. Then, $\mathcal{C}$ is asymptotically optimal GUCI.
\end{theorem}
\begin{proof}
From Theorem~\ref{thm2}, we know that $\mathcal{C}=(\mathcal{D}_{RLE},\varphi_{\psi})$ is GUCI.
Due to Theorem~\ref{thm3} and
\[
   \lim\limits_{H(P)\to +\infty}R_{\psi}(H(P))=1,
\]
we obtain
\[
 \lim\limits_{H(P)\to +\infty}T_{\mathcal{C}}(H(P))=1.
\]
Therefore, $\mathcal{C}$ is asymptotically optimal GUCI.
\end{proof}

\section{Conclusions}\label{sec_con}
In this paper, GUCI is proposed to solve the issue of UCI that the ratio of the expected codeword length to $H(P)$ cannot be within a constant factor $K$ when $H(P)$ is extremely small.
We construct a VV code $\mathcal{C}=(\mathcal{D}_{RLE},\varphi_{\psi})$ through RLE and UCI $\psi$,
and we proved that $\mathcal{C}$ is GUCI or asymptotically optimal GUCI when UCI $\psi$ satisfies certain conditions.
We propose a class of GUCI $\mathcal{C}=(\mathcal{D}_{RLE},\varphi_{\gamma})$ to achieve the expansion factor $K_{\mathcal{C}}=2$ and show that the optimal GUCI is in the range $1\leq K_{\mathcal{C}}^{*}\leq 2$.
GUCI is suitable for small entropy. For example, in image compression, the frequency-domain coefficients have many zeros after
the quantization process~\cite{2015A}.
Furthermore, when the entropy is very large or $P(0)$ is not large, it is still possible that the coding rate $R_{\mathcal{C}}$ is less than the expected codeword length $E_{P}(L_{\psi})$.

\appendix
The proof of Lemma~\ref{lemma1} given in the Appendix is an extension of the proof in~\cite{FK72,LT91}.
Before concrete proof, let us give some definitions.
Suppose $\mathcal{D}$ is a  dictionary. For every $\alpha\in\mathcal{D}$, then $\mathcal{D}[\alpha]\triangleq (\mathcal{D}-\{\alpha\})\cup\alpha\mathcal{A}$ is also a dictionary,
where $\alpha\mathcal{A}\triangleq\{\alpha\beta\mid\beta\in\mathcal{A}\}$.
Dictionary $\mathcal{D}[\alpha]$ is said to be an extension of dictionary $\mathcal{D}$, and $\alpha$ is termed the extending string.
Obviously, when $\mathcal{D}$ is proper and complete, then $\mathcal{D}[\alpha]$ is also proper and complete.
Let $\mathcal{D}_{n}$ denote a proper and complete dictionary as follows.
\[
\mathcal{D}_{n}\triangleq
\{ \alpha\in\mathcal{D} \mid |\alpha|<n \}\cup \mathcal{D}_{n}^{\perp},
\]
where the length of the elements in $\mathcal{D}^{\perp}$ are all $n\in\mathbb{N}^{+}$, and $\mathcal{D}_{n}^{\perp}$ makes the dictionary $\mathcal{D}_{n}$ is proper and complete.
In particular, $\mathcal{D}_{1}=\mathcal{A}$.
\begin{lemma}(Lemma~\ref{lemma1} Restated)
	Let $S=(P,\mathcal{A})$ denote a discrete memoryless source with entropy $H(P)<\infty$ and a countable alphabet $\mathcal{A}$.
	Given a VV code $\mathcal{C}$ with a proper and ASC dictionary $\mathcal{D}$, than
	\begin{equation}\label{e11}
		H(\mathcal{D})= H(P)\overline{l(\mathcal{D})},
	\end{equation}
	where $H(\mathcal{D})=-\sum_{\alpha\in\mathcal{D}}P(\alpha)\log_{2}P(\alpha)$ and $\overline{l(\mathcal{D})}=\sum_{\alpha\in\mathcal{D}}P(\alpha)|\alpha|$.
\end{lemma}
\begin{proof}
The proof is divided into three parts.
First, we prove that if the dictionary satisfies equation \eqref{e11}, then the dictionary after finite extensions also satisfies equation \eqref{e11}.
Next, the following equation will be proved.
\begin{equation}\label{e12}
	H(\mathcal{D}_{n})= H(P)\overline{l(\mathcal{D}_{n})},
\end{equation}
for all $n\in\mathbb{N}^{+}$.
Finally, the proof for equation \eqref{e11} will be presented.
\begin{enumerate}
	\item Suppose $\mathcal{S}$ is a  dictionary.
	By recursion, we only need to prove that when $\mathcal{S}$ satisfies equation \eqref{e11}, then $\mathcal{S}[\alpha]$ after one extension also satisfies equation \eqref{e11}.
	We obtain
	\begin{equation*}
		\begin{aligned}
			\overline{l(\mathcal{S}[\alpha])}&= \sum_{\beta\in\mathcal{S}-\{\alpha\}}P(\beta)|\beta|+ \sum_{\beta\in\alpha\mathcal{A}}P(\beta)|\beta|                             \\
			&= \sum_{\beta\in\mathcal{S}}P(\beta)|\beta|-P(\alpha)|\alpha|+ P(\alpha)(|\alpha|+1)     \\
			&= \overline{l(\mathcal{S})}+P(\alpha),
		\end{aligned}
	\end{equation*}
	and
	\begin{equation*}
		\begin{aligned}
			H(\mathcal{S}[\alpha])&= -\sum_{\beta\in\mathcal{S}-\{\alpha\}}P(\beta)\log_{2}P(\beta)-\sum_{\beta\in\alpha\mathcal{A}}P(\beta)\log_{2}P(\beta)                                \\
			&= -\sum_{\beta\in\mathcal{S}}P(\beta)\log_{2}P(\beta)+P(\alpha)\log_{2}P(\alpha)-\sum_{\beta\in\mathcal{A}}P(\alpha)P(\beta)\log_{2}P(\alpha)P(\beta)      \\
			&= H(\mathcal{S})+P(\alpha)\log_{2}P(\alpha)-P(\alpha)\log_{2}P(\alpha)-P(\alpha)\sum_{\beta\in\mathcal{A}}P(\beta)\log_{2}P(\beta)      \\
			&= H(P)\overline{l(\mathcal{S})}+H(P)P(\alpha)    \\
			&= H(P)\overline{l(\mathcal{S}[\alpha])}.
		\end{aligned}
	\end{equation*}
	The first part of the proof is complete.
	\item
	We prove equation \eqref{e12} by mathematical induction.
	When $n=1$, we have $\mathcal{D}_{1}=\mathcal{A}$ and
	\[
	H(P)\overline{l(\mathcal{D}_{1})}=H(P)\times1=H(\mathcal{D}_{1}).
	\]
	Suppose equation \eqref{e12} holds when $n=m$.
	Now, we consider the extension process from $\mathcal{D}_{m}$ to $\mathcal{D}_{m+1}$.
	Let
	\[
	T\triangleq\{\alpha\in\mathcal{D}_{m}^{\perp}\mid \alpha \notin\mathcal{D} \}.
	\]
	Due to the definition of $\mathcal{D}_{m}$, we know that $T$ is exactly the set of all extending strings from $\mathcal{D}_{m}$ to $\mathcal{D}_{m+1}$.
	If $|T|<\infty$, then $\mathcal{D}_{m+1}$ is obtained by $\mathcal{D}_{m}$ after finite extensions.
	We obtain $H(\mathcal{D}_{m+1})= H(P)\overline{l(\mathcal{D}_{m+1})}$ due to the first part of the proof.
	If $|T|=\infty$, because $\mathcal{A}$ is countable and the length of the elements in $T$ are all $m$, $T$ is also countable.
	Therefore, suppose $T\triangleq\{\alpha_{i}\}_{i=1}^{\infty}$, the extension process from $\mathcal{D}_{m}$ to $\mathcal{D}_{m+1}$ is as follows.
	\begin{equation*}
		\begin{aligned}
			\mathcal{D}_{m+1,1}\triangleq &(\mathcal{D}_{m}-\{\alpha_{1}\})\cup \alpha_{1}\mathcal{A},                     \\
			\mathcal{D}_{m+1,2}\triangleq & (\mathcal{D}_{m+1,1}-\{\alpha_{2}\})\cup \alpha_{2}\mathcal{A}= (\mathcal{D}_{m}-\{\alpha_{i}\}_{i=1}^{2})\cup \{\alpha_{i}\mathcal{A}\}_{i=1}^{2},     \\
			\vdots &  \\
			\mathcal{D}_{m+1,k} \triangleq &  (\mathcal{D}_{m+1,k-1}-\{\alpha_{k}\})\cup \alpha_{k}\mathcal{A}= (\mathcal{D}_{m}-\{\alpha_{i}\}_{i=1}^{k})\cup \{\alpha_{i}\mathcal{A}\}_{i=1}^{k},     \\
			\vdots &
		\end{aligned}
	\end{equation*}
	We have the following three equations.
	\begin{enumerate}
		\item[(\romannumeral1)] $H(\mathcal{D}_{m+1,k})= H(P)\overline{l(\mathcal{D}_{m+1,k})}$, for all $k\in\mathbb{N}^{+}$.
		\item[(\romannumeral2)] $\lim\limits_{k\to +\infty}\mathcal{D}_{m+1,k}=\mathcal{D}_{m+1}$.
		\item[(\romannumeral3)] $\mathcal{D}_{m+1}=(\mathcal{D}_{m}-\{\alpha_{i}\}_{i=1}^{\infty})\cup \{\alpha_{i}\mathcal{A}\}_{i=1}^{\infty}$.
	\end{enumerate}
	Next, we prove the following two equations.
	\begin{equation}\label{e13}
		\begin{aligned}
			\lim\limits_{k\to +\infty}\overline{l(\mathcal{D}_{m+1,k})}&=\overline{l(\mathcal{D}_{m+1})}.            \\
			\lim\limits_{k\to +\infty} H(\mathcal{D}_{m+1,k})&=H(\mathcal{D}_{m+1}).
		\end{aligned}
	\end{equation}
	First, we have
	\begin{equation*}
		\begin{aligned}
			\overline{l(\mathcal{D}_{m+1})}& = \sum_{\alpha\in\mathcal{D}_{m+1}}P(\alpha)|\alpha|                               \\
			&\geq \sum_{\alpha\in\mathcal{D}_{m+1,k}}P(\alpha)|\alpha|       \\
			& = \overline{l(\mathcal{D}_{m+1,k})}      \\
			&\geq \sum_{\alpha\in(\mathcal{D}_{m}-\{\alpha_{i}\}_{i=1}^{\infty})\cup\{\alpha_{i}\mathcal{A}\}_{i=1}^{k}}P(\alpha)|\alpha|.
		\end{aligned}
	\end{equation*}
	Taking the limit $k\rightarrow\infty$, we obtain
	\begin{equation*}
		\begin{aligned}
			\overline{l(\mathcal{D}_{m+1})}&\geq \lim\limits_{k\to +\infty} \overline{l(\mathcal{D}_{m+1,k})}     \\
			&\geq \lim\limits_{k\to +\infty} \sum_{\alpha\in(\mathcal{D}_{m}-\{\alpha_{i}\}_{i=1}^{\infty})\cup\{\alpha_{i}\mathcal{A}\}_{i=1}^{k}}P(\alpha)|\alpha|  \\
			&=\overline{l(\mathcal{D}_{m+1})}.
		\end{aligned}
	\end{equation*}
	Then, we have
	\begin{equation*}
		\begin{aligned}
			H(\mathcal{D}_{m+1})&= -\sum_{\alpha\in(\mathcal{D}_{m}-\{\alpha_{i}\}_{i=1}^{\infty})\cup \{\alpha_{i}\mathcal{A}\}_{i=1}^{k}}P(\alpha)\log_{2}P(\alpha)-\sum_{\alpha\in\{\alpha_{i}\mathcal{A}\}_{i=k+1}^{\infty}}P(\alpha)\log_{2}P(\alpha)       \\
			&\overset{(a)}{\geq}-\sum_{\alpha\in(\mathcal{D}_{m}-\{\alpha_{i}\}_{i=1}^{\infty})\cup \{\alpha_{i}\mathcal{A}\}_{i=1}^{k}}P(\alpha)\log_{2}P(\alpha)-\sum_{\alpha\in\{\alpha_{i}\}_{i=k+1}^{\infty}}P(\alpha)\log_{2}P(\alpha)       \\
			&=-\sum_{\alpha\in(\mathcal{D}_{m}-\{\alpha_{i}\}_{i=1}^{k})\cup \{\alpha_{i}\mathcal{A}\}_{i=1}^{k}}P(\alpha)\log_{2}P(\alpha)       \\
			&= H(\mathcal{D}_{m+1,k})    \\
			&\geq -\sum_{\alpha\in(\mathcal{D}_{m}-\{\alpha_{i}\}_{i=1}^{\infty})\cup \{\alpha_{i}\mathcal{A}\}_{i=1}^{k}}P(\alpha)\log_{2}P(\alpha),
		\end{aligned}
	\end{equation*}
	where $(a)$ is due to $-\sum_{\alpha\in\alpha_{i}\mathcal{A}}P(\alpha)\log_{2}P(\alpha)\geq -P(\alpha_{i})\log_{2}P(\alpha_{i})$, for all $i\in\mathbb{N}^{+}$.
	Taking the limit $k\rightarrow\infty$, we obtain
	\begin{equation*}
		\begin{aligned}
			H(\mathcal{D}_{m+1})&\geq \lim\limits_{k\to +\infty} H(\mathcal{D}_{m+1,k})    \\
			&\geq \lim\limits_{k\to +\infty} -\sum_{\alpha\in(\mathcal{D}_{m}-\{\alpha_{i}\}_{i=1}^{\infty})\cup \{\alpha_{i}\mathcal{A}\}_{i=1}^{k}}P(\alpha)\log_{2}P(\alpha) \\
			&=H(\mathcal{D}_{m+1}).
		\end{aligned}
	\end{equation*}
	Equation \eqref{e13} is proved. From the perspective of mathematical analysis, equation \eqref{e13} essentially considers whether the function and the limit can be exchanged.
	For example, $\lim\limits_{k\to +\infty} H(\mathcal{D}_{m+1,k})=H(\lim\limits_{k\to +\infty}\mathcal{D}_{m+1,k})$.
	Finally, form equation \eqref{e13}, we have
	\begin{equation*}
		\begin{aligned}
			H(\mathcal{D}_{m+1})&= \lim\limits_{k\to +\infty} H(\mathcal{D}_{m+1,k})       \\
			&= \lim\limits_{k\to +\infty} H(P)\overline{l(\mathcal{D}_{m+1,k})}   \\
			&= H(P)\overline{l(\mathcal{D}_{m+1})} .
		\end{aligned}
	\end{equation*}
	The second part of the proof is complete.
	\item
	We prove the following two equations similar to equation \eqref{e13}.
	\begin{equation}\label{e14}
		\begin{aligned}
			\lim\limits_{m\to +\infty}\overline{l(\mathcal{D}_{m})}&=\overline{l(\mathcal{D})}.            \\
			\lim\limits_{m\to +\infty} H(\mathcal{D}_{m})&=H(\mathcal{D}).
		\end{aligned}
	\end{equation}
	First, we have
	\begin{equation*}
		\begin{aligned}
			\overline{l(\mathcal{D})}& = \sum_{\alpha\in\mathcal{D}}P(\alpha)|\alpha|                               \\
			&\geq \sum_{\alpha\in\mathcal{D},|\alpha|<m}P(\alpha)|\alpha|+ \sum_{\alpha\in\mathcal{D}_{m}^{\perp}}P(\alpha)|\alpha|      \\
			& = \overline{l(\mathcal{D}_{m})}      \\
			&\geq \sum_{\alpha\in\mathcal{D},|\alpha|<m}P(\alpha)|\alpha|.
		\end{aligned}
	\end{equation*}
	Taking the limit $m\rightarrow\infty$, we obtain
	\begin{equation*}
		\begin{aligned}
			\overline{l(\mathcal{D})}&\geq \lim\limits_{m\to +\infty} \overline{l(\mathcal{D}_{m})}     \\
			&\geq \lim\limits_{m\to +\infty}\sum_{\alpha\in\mathcal{D},|\alpha|<m}P(\alpha)|\alpha|  \\
			&=\overline{l(\mathcal{D})}.
		\end{aligned}
	\end{equation*}
	Then, we have
	\begin{equation*}
		\begin{aligned}
			H(\mathcal{D})&= -\sum_{\alpha\in\mathcal{D},|\alpha|<m}P(\alpha)\log_{2}P(\alpha)-\sum_{\alpha\in\mathcal{D},|\alpha|\geq m}P(\alpha)\log_{2}P(\alpha)       \\
			& \overset{(a)}{\geq}-\sum_{\alpha\in\mathcal{D},|\alpha|<m}P(\alpha)\log_{2}P(\alpha)-\sum_{\beta\in\mathcal{D}^{\perp}}P(\beta)\log_{2}P(\beta)        \\
			&= H(\mathcal{D}_{m})    \\
			&\geq -\sum_{\alpha\in\mathcal{D},|\alpha|<m}P(\alpha)\log_{2}P(\alpha),
		\end{aligned}
	\end{equation*}
	where $(a)$ is due to
	\[
	-\sum_{\substack{\hbox{$\beta$ is the prefix of $\alpha$},\\\alpha\in\mathcal{D}}}P(\alpha)\log_{2}P(\alpha)\geq -P(\beta)\log_{2}P(\beta),
	\]
	for all $\beta\in\mathcal{D}^{\perp}$.
	Taking the limit $m\rightarrow\infty$, we obtain
	\begin{equation*}
		\begin{aligned}
			H(\mathcal{D})&\geq \lim\limits_{m\to +\infty} H(\mathcal{D}_{m})    \\
			&\geq \lim\limits_{m\to +\infty} -\sum_{\alpha\in\mathcal{D},|\alpha|<m}P(\alpha)\log_{2}P(\alpha) \\
			&=H(\mathcal{D}).
		\end{aligned}
	\end{equation*}
	Equation \eqref{e14} is proved.
	Finally, form equation \eqref{e14}, we have
	\begin{equation*}
		\begin{aligned}
			H(\mathcal{D})&= \lim\limits_{m\to +\infty} H(\mathcal{D}_{m})       \\
			&= \lim\limits_{m\to +\infty} H(P)\overline{l(\mathcal{D}_{m})}   \\
			&= H(P)\overline{l(\mathcal{D})} .
		\end{aligned}
	\end{equation*}
\end{enumerate}
The proof is completed.
\end{proof}
\bibliographystyle{IEEEtran}
\bibliography{IEEEabrv,refs}
\end{document}